\documentclass[12pt, leqno]{article}

\usepackage{amssymb,amsfonts,amsthm}
\usepackage{graphicx}
\usepackage{subfigure}
\usepackage{amssymb}
\usepackage{amsthm}
\usepackage{amsmath}
\usepackage{bm}             
\usepackage[mathscr]{eucal}
\usepackage{color}
\usepackage{cancel}
\usepackage{enumerate}
\usepackage{psfrag}
\usepackage{appendix}

\newcommand{\eps}{\epsilon}

\newcommand{\mcR}{\mathcal{R}}

\newcommand{\mcW}{\mathcal{W}}
\newcommand{\mcL}{\mathcal{L}}

\newtheorem{theorem}{Theorem}

\newtheorem{lemma}[theorem]{Lemma}

\theoremstyle{definition}

\usepackage{fullpage}

\title{\Huge{Concentrating solutions of the relativistic Vlasov-Maxwell system}}
\author{Jonathan Ben-Artzi\\ {\small School of Mathematics}\\ {\small Cardiff University}\\ {\small Cardiff, United Kingdom}\\ {\small \tt Ben-ArtziJ@cardiff.ac.uk} \\[1cm]
Simone Calogero\\{\small Department of Mathematics}\\ {\small Chalmers Institute of Technology, University of Gothenburg}\\ {\small Gothenburg, Sweden}\\ {\small \tt calogero@chalmers.se}\\[1cm]
Stephen Pankavich \\ {\small Department of Applied Mathematics and Statistics}\\ {\small Colorado School of Mines}\\ {\small Golden, CO USA}\\ {\small \tt pankavic@mines.edu} }
\date{\today}

\begin{document}
\maketitle

\begin{abstract}
We study smooth, global-in-time solutions of the relativistic Vlasov-Maxwell system that possess arbitrarily large charge densities and electric fields. In particular, we construct spherically symmetric solutions that describe a thin shell of equally charged particles concentrating arbitrarily close to the origin and which give rise to charge densities and electric fields as large as one desires at some finite time. We show that these solutions exist even for arbitrarily small initial data or any desired mass. In the latter case, the time at which solutions concentrate can also be made arbitrarily large.
\end{abstract}

\section{Introduction}
Let $f(t,x,p)\geq0$ denote the one particle distribution in phase space of a monocharged plasma. Taking relativistic effects into account, but neglecting collisions among the particles, $f$ satisfies the relativistic Vlasov-Maxwell system:
\begin{equation}
\tag{RVM}
\label{RVM}
\left \{ \begin{gathered}
\partial_{t}f+\hat{p} \cdot\nabla_{x}f+(E + \hat{p} \wedge B) \cdot\nabla_{p}f=0\\
\partial_{t} E=\nabla \wedge B- 4\pi j, \quad \nabla \cdot E=4\pi\rho,\\
\partial_{t} B=-\nabla \wedge E, \quad\nabla \cdot B=0.
\end{gathered} \right .
\end{equation}
where
\begin{equation}
\label{sources}
\rho(t,x)=\int_{\mathbb{R}^3} f(t,x,p)\,dp, \quad j(t,x)=\int_{\mathbb{R}^3} \hat{p} f(t,x,p)\,dp
\end{equation}
are the charge and current density of the plasma, while
\begin{equation}
\label{phat}
\hat{p} = \frac{p}{\sqrt{1 + \vert p \vert^2}}
\end{equation}
is the relativistic velocity. Additionally, $E(t,x)$ and $B(t,x)$ are the self-consistent electric and magnetic fields generated by the charged particles, and we have chosen units such that the mass and charge of each particle, as well as the speed of light, are normalized to one.

Under the assumption of spherically-symmetric initial data, it is well-known \cite{GS, Horst} that solutions of \eqref{RVM} exist and remain spherically symmetric for all time. Additionally, the corresponding magnetic field is constant (and equal to zero for finite energy solutions), and therefore  \eqref{RVM} reduces to the relativistic Vlasov-Poisson system:
 \begin{equation}
\tag{RVP}
\label{RVP}
\left \{ \begin{aligned}
& \partial_{t}f+\hat{p}\cdot\nabla_{x}f+E \cdot\nabla_{p}f=0\\
& E(t,x) = \int_{\mathbb{R}^3} \frac{x-y}{\vert x - y \vert^3} \rho(t,y) \ dy.
\end{aligned} \right .
\end{equation}
It should be noted that global existence of solutions to \eqref{RVM} without spherically-symmetric initial data remains an extremely challenging, unsolved question - see~\cite{Glassey} for a background on the Cauchy problem for~\eqref{RVM}. 
Here, we also consider the Cauchy problem and therefore require given initial data
$$f(0,x,p) = f_0(x,p) \geq 0$$
that is spherically-symmetric.
Because of this symmetry, no initial field values are needed to complete the description of the system. 

In the present paper, we prove the existence of spherically symmetric solutions of \eqref{RVM} (or, equivalently of~\eqref{RVP}) which give rise to charge densities and electric fields that become arbitrarily large at some finite time, even if initially they are chosen to be arbitrarily small. We prove two versions of this result. More specifically, our first theorem shows that one may construct solutions of \eqref{RVM} whose density and field are initially as small as desired, but which become as large as desired at some later time.
\begin{theorem}
\label{T1}
For any constants $C_1, C_2 > 0$ there exists a smooth, spherically symmetric solution of \eqref{RVM} such that
$$ \Vert \rho(0) \Vert_\infty, \quad \Vert E(0) \Vert_\infty  \leq C_1$$
but for some time $T > 0$,
$$ \Vert \rho(T) \Vert_\infty, \quad \Vert E(T) \Vert_\infty  \geq C_2.$$
\end{theorem}

The next main theorem shows that one may construct solutions possessing any desired mass, but whose density and field become arbitrarily large at any given time.
\begin{theorem}
\label{T2}
For any constants $C_1, C_2> 0$ and any $T > 0$ there exists a smooth, spherically symmetric solution of \eqref{RVM} such that
$$ M = \iint_{\mathbb{R}^6} f_0(x,p) \ dp dx = C_1$$
and
$$ \Vert \rho(T) \Vert_\infty, \quad  \Vert E(T) \Vert_\infty  \geq C_2.$$
\end{theorem}
These results are somewhat surprising as the arbitrarily large density and field values that can arise at any prescribed time seem to contradict the general intuition that spherically-symmetric solutions are among the most well-behaved of any launched by the system.
In particular, Theorem~\ref{T2} complements a previously-established decay theorem \cite{Horst}, which states that the density and field generated by any spherically-symmetric solution of \eqref{RVM} must obey sharp asymptotic decay estimates for $t$ suitably large. 
Namely, H\"{o}rst proves that there is $C> 0$ and $T>0$ sufficiently large such that
$$\Vert \rho(t) \Vert_\infty \leq Ct^{-2}, \qquad \Vert E(t) \Vert_\infty \leq Ct^{-3}$$
for all $t \geq T$.
Our second theorem, on the other hand, demonstrates that the time needed for solutions to transition from their intermediate asymptotic behavior, during which they may attain large values, to their final asymptotic behavior can be made as large as one desires, even if the total mass is taken to be small.
In addition to H\"{o}rst, other authors have studied the large time asymptotic behavior of \eqref{RVM} and related kinetic models of plasma. In this direction, we mention \cite{GPS, GPS2, GPS3, GPS4, IR, Yang, Young, Young2}.

Previously, we proved \cite{BCP} an analogous version of these theorems for the classical limit (i.e., as the speed of light $c \to \infty$) of \eqref{RVM}, namely the well-known (non-relativistic) Vlasov-Poisson system. With the introduction of relativistic velocities - and hence a uniform velocity bound - the physical nature of the problem is fundamentally different, and particle velocities can no longer be taken arbitrarily large. 
However, we develop new energy estimates along characteristics (Lemma \ref{L1}) that express the particle trajectories in terms of their relativistic free-streaming counterparts, and thus utilize a novel argument within the proof of Theorems \ref{T1} and \ref{T2}. 

We conclude this introduction by rewriting the spherically-symmetric, relativistic Vlasov-Maxwell system in spherical coordinates. Defining the spatial radius, radial velocity, and square of the angular momentum by
\begin{equation}
\label{ang}
r = \vert x \vert, \qquad w = \frac{x \cdot p}{r}, \qquad \ell = \vert x \times p \vert^2,
\end{equation}
the particle distribution may be written as $f = f(t,r,w,\ell)$ and satisfies the reduced Vlasov equation
\begin{equation}
\label{vlasovang}
\partial_{t}f+\frac{w}{\sqrt{1 + w^2+ \ell r^{-2}}}\partial_r f+\left ( \frac{\ell}{r^3\sqrt{1 + w^2+ \ell r^{-2}}} + \frac{m(t,r)}{r^2} \right ) \partial_w f=0
\end{equation}
where
\begin{equation}
\label{massang}
m(t,r) = 4\pi \int_0^r s^2 \rho(t,s) \ ds
\end{equation}
and 
\begin{equation}
\label{rhoang}
\rho(t,r) = \frac{\pi}{r^2} \int_0^\infty \int_{-\infty}^{\infty} f(t,r,w,\ell) \ dw \ d\ell.
\end{equation}
The electric field is given by the expression
\begin{equation}
\label{fieldang}
E(t,x) = \frac{m(t,r)}{r^2} \frac{x}{r}.
\end{equation}
As previously mentioned, the magnetic field can be taken to be identically zero, and the current is not required to describe the system.
Finally, the total mass can be expressed as
$$M = 4\pi^2 \int_0^\infty \int_{-\infty}^{\infty} \int_0^\infty f_0(r, w, \ell) \ d\ell dw dr.$$
Whenever necessary, we will abuse notation so as to use both Cartesian and angular coordinates to refer to functions;  for instance the particle density $f$ will be written both as $f(t,x,p)$ and $f(t,r,w,\ell)$ when needed.

The forward characteristics of the Vlasov equation~\eqref{vlasovang} are the solutions of 
\begin{equation}
\label{charang}
\left \{
\begin{aligned}
&\dot{\mcR}(s)=\frac{\mcW(s)}{\sqrt{1 + \mcW(s)^2 + \mcL(s) \mcR(s)^{-2}}},\\
&\dot{\mcW}(s)= \frac{\mcL(s)}{\mcR(s)^3 \sqrt{1 + \mcW(s)^2 + \mcL(s) \mcR(s)^{-2}}} + \frac{m(s, \mcR(s))}{\mcR(s)^2},\\
&\dot{\mcL}(s)= 0.
\end{aligned}
\right.
\end{equation}
for $s\geq0$, subject to the initial conditions
\begin{equation}
\label{charanginit}
\mcR(0) = r, \qquad \mcW(0) = w, \qquad \mcL(0) = \ell.
\end{equation}
In particular, because the angular momentum of particles is conserved in time on the support of $f(t)$, we have $\mcL(s) = \ell$ for every $s \geq 0$.
Throughout, we will estimate particle behavior on the support of $f$, and thus for convenience we define for all $t \geq 0$
$$S(t) = \{ (r,w,\ell) : f(t,r,w,\ell) > 0\}$$
so that, in particular
$$S(0) = \{ (r,w,\ell) : f_0(r,w,\ell) > 0\}.$$

The remainder of the paper is structured as follows. The proofs of Theorems~\ref{T1} and \ref{T2} are contained within Section~\ref{proofs}, while Section~\ref{characteristics} is devoted to an important lemma that is crucial to understanding the behavior of the particle characteristics. 
Additionally, we remark that a theorem similar to these, but utilizing initial data that is not spherically-symmetric can also be established using our methods. In particular, if one chooses the same spherically-symmetric initial data that we construct in the proofs of the aforementioned theorems, but does so only within a ball around the origin, and then allows for non-symmetric data outside of the domain of dependence of the corresponding solution, then the identical concentration results must follow. 
Throughout the paper $C$ will represent a constant that may change from line to line, but when necessary to denote a certain constant, we will distinguish this value with a subscript, e.g. $C_0$.

\section{Behavior of the Characteristics}
\label{characteristics}
We first study the behavior of the characteristic system \eqref{charang} and relate the solutions to their initial positions and momenta.
In particular, we use the convex nature of particle characteristics to estimate both their minimal value and the corresponding time at which it is achieved.

\begin{lemma}
\label{L1}
Let $r, \ell > 0$ and $w < 0$ be given, let $(\mcR(t), \mcW(t), \ell)$ satisfy \eqref{charang} and \eqref{charanginit} for all $t \geq 0$,
and define
$$D = \ell + Mr\sqrt{1 + w^2 + \ell r^{-2}}.$$
Then, we have the following: 
\begin{enumerate}
\item There exists a unique $T_0 > 0$ such that
$$\mcW(t)  < 0 \ \mathrm{for} \ t \in [0,T_0),$$ 
$$\mcW(T_0) =0, \ \mathrm{and}$$
$$\mcW(t)  > 0 \ \mathrm{for} \ t \in (T_0,\infty).$$

\item Further, $T_0$ satisfies the bounds
$$r \left (1 - \sqrt{\frac{D}{r^2w^2 + D}} \right) \leq T_0 \leq \frac{-w r^3\sqrt{1 + w^2 + \ell r^{-2}}}{\ell}.$$

\item Define 
\begin{equation}
\label{R_-}
\mcR_- = r\sqrt{\frac{\ell}{r^2w^2 + \ell}}
\end{equation}
and 
\begin{equation}
\label{R_+}
\mcR_+ = r\sqrt{\frac{D}{r^2w^2 + D}}.
\end{equation}
Then, 
$$\mcR_- \leq \mcR(T_0) \leq \mcR_+.$$

\item For all $t \in [0,T_0]$, we have
\begin{equation}
\label{energy_lower}
\mcW(t)^2 + \ell \mcR(t)^{-2} \leq w^2 + \ell r^{-2}
\end{equation}

\item For all $t \in [0,T_0]$, we have
\begin{equation}
\label{energy_upper}
\mcR(t)^2 \leq \left (r - \frac{\vert w \vert}{\sqrt{1 + w^2 + \ell r^{-2}}}t \right )^2 + \frac{D}{r^2(1 + w^2 + \ell r^{-2})}t^2.
\end{equation}

\end{enumerate}
\end{lemma}

\begin{proof}
To begin, define
$$T_0 = \sup \{t \geq 0 : \mcW(t) \leq 0 \}$$
and note that $w < 0$ implies $T_0 > 0$.
In addition, because
$$\dot{\mcR}(t)  = \frac{\mcW(t)}{\sqrt{1 + \mcW(t)^2 + \ell \mcR(t)^{-2}}} \leq 0$$
on $[0,T_0]$, we have
\begin{equation}
\label{Rdec}
\mcR(t) \leq r \quad \mathrm{ \ for \ all \ } t \in [0,T_0].
\end{equation}

We first establish energy estimates on $[0,T_0]$. Taking the derivative of the particle kinetic energy along characteristics, a brief calculation yields 
$$\frac{d}{dt} \sqrt{1 + \mcW(t)^2 + \ell\mcR(t)^{-2}} = m(t,\mcR(t)) \mcR(t)^{-2} \frac{\mcW(t)}{\sqrt{1  + \mcW(t)^2 + \ell\mcR(t)^{-2}}}.$$
Since $0 \leq m(t, \mcR(t)) \leq M$ and $\mcW(t) \leq 0$, we find for $t \in [0,T_0]$ both
\begin{equation}
\label{energy2}
\frac{d}{dt} \sqrt{1 + \mcW(t)^2 + \ell\mcR(t)^{-2}} \leq 0
\end{equation}
and
\begin{equation}
\label{energy3}
\frac{d}{dt} \sqrt{1 + \mcW(t)^2 + \ell\mcR(t)^{-2}} \geq M\mcR(t)^{-2} \dot{\mcR}(t).
\end{equation}
From \eqref{energy2}, the particle kinetic energy is decreasing on this interval and hence
$$\sqrt{1 + \mcW(t)^2 + \ell\mcR(t)^{-2}} \leq \sqrt{1 + w^2  + \ell r^{-2}}$$
for all $t \in [0,T_0]$, which establishes \eqref{energy_lower}.
Additionally, evaluating \eqref{energy_lower} at $t = T_0$ and using $\mcW(T_0) = 0$ produces $\mcR(T_0) \geq \mcR_-$ with $\mcR_-$ given by \eqref{R_-}.

Returning to \eqref{energy3}, this can be rewritten as
$$\frac{d}{dt} \left (\sqrt{1 + \mcW(t)^2 + \ell\mcR(t)^{-2}} + M\mcR(t)^{-1} \right ) \geq 0$$
and upon integrating over $[0,t]$, yields
$$\sqrt{1 + \mcW(t)^2 + \ell\mcR(t)^{-2}} - \sqrt{1 + w^2  + \ell r^{-2}} \geq -M \left (\mcR(t)^{-1} - r^{-1} \right ).$$ 
Multiplying both sides of the inequality by the conjugate and using \eqref{energy_lower}, we arrive at
$$\mcW(t)^2 - w^2 + \ell(\mcR(t)^{-2} - r^{-2}) \geq -2M \sqrt{1 + w^2  + \ell r^{-2}}\left (\mcR(t)^{-1} - r^{-1} \right ).$$
By \eqref{Rdec} it follows that $ \mcR(t)^{-1} + r^{-1}\geq 2 r^{-1}$,
and thus we find
\begin{eqnarray*}
\mcW(t)^2 & \geq & w^2 - \ell \left (\mcR(t)^{-2} - r^{-2} \right) -2M\sqrt{1 + w^2  + \ell r^{-2}}\left (\mcR(t)^{-1} -r^{-1} \right)\\
& = & w^2 - \left (\ell + \frac{2M\sqrt{1 + w^2  + \ell r^{-2}}}{\mcR(t)^{-1} + r^{-1}} \right) \left (\mcR(t)^{-2} - r^{-2} \right )\\
& \geq & w^2 - (\ell + Mr\sqrt{1 + w^2  + \ell r^{-2}}) \left (\mcR(t)^{-2} - r^{-2} \right )\\
& = & w^2 - D \left (\mcR(t)^{-2} - r^{-2} \right )
\end{eqnarray*}
for all $t \in [0,T_0]$.
Evaluating this inequality at $t = T_0$ and rearranging yields $\mcR(T_0) \leq \mcR_+$ with $\mcR_+$ given by \eqref{R_+}.

Returning to the lower bound for $\mcW(t)^2$ and using \eqref{energy_lower} within \eqref{charang}, we find for $t \in [0,T_0]$
\begin{eqnarray*}
\vert \dot{\mcR}(t) \vert^2 & = & \frac{\mcW(t)^2}{1 + \mcW(t)^2 + \ell \mcR(t)^{-2}}\\
& \geq & \frac{w^2 - D \left (\mcR(t)^{-2} - r^{-2} \right )}{1 + w^2 + \ell r^{-2}}.
\end{eqnarray*}
Let $$A = \frac{w^2 + Dr^{-2}}{1 + w^2 + \ell r^{-2}} \qquad \mathrm{and} \qquad B = \frac{D}{1 + w^2 + \ell r^{-2}}$$
then multiply this inequality by $\mcR(t)^2$ to find
\begin{equation}
\label{R2}
\left \vert \frac{1}{2} \frac{d}{dt} \left ( \mcR(t)^2 \right ) \right \vert^2 \geq A\mcR(t)^2 - B.
\end{equation}
If
\begin{equation}
\label{lb}
\mcR(t) > \sqrt{\frac{B}{A}} = \sqrt{\frac{D}{w^2 + Dr^{-2}}} = \mcR_+
\end{equation}
then the right side of \eqref{R2} is positive, and because $\dot{\mcR}(t) \leq 0$ this yields
$$ \frac{-\frac{1}{2} \frac{d}{dt} \vert \mcR(t) \vert^2}{\sqrt{A \mcR(t)^2 - B}} \geq 1$$
and thus
$$\frac{d}{dt} \sqrt{A\mcR(t)^2 - B} \leq -A$$
for $t \in [0,T_0]$.
Integrating over $[0,t]$ and rearranging terms, we find
$$\mcR(t)^2 \leq r^2 - 2t\sqrt{Ar^2 - B} + At^2 = \left (r - \sqrt{A-Br^{-2}}t \right )^2 + Br^{-2}t^2,$$
which becomes \eqref{energy_upper} upon inserting the formulas for $A$ and $B$.
Thus, combining this estimate with \eqref{lb} we find
$$\mcR(t)^2 \leq \max \left \{ \mcR_+^2, \left (r - \frac{\vert w \vert}{\sqrt{1 + w^2 + \ell r^{-2}}}t \right )^2 + \frac{D}{r^2(1 + w^2 + \ell r^{-2})}t^2 \right\}$$
for all $t \in [0,T_0]$.
Finally, the upper bound in this estimate arising from condition \eqref{lb} can be removed by noting that $\mcR_+^2 = \frac{D}{w^2 + Dr^{-2}}$ is, in fact, the minimum of the parabola in $t$, which occurs at the time $$t_{min} = \frac{r\vert w\vert}{w^2 +Dr^{-2}}\sqrt{1 + w^2 + \ell r^{-2}}.$$ 
This implies
$$\mcR_+^2 \leq \left (r - \frac{\vert w \vert}{\sqrt{1 + w^2 + \ell r^{-2}}}t \right )^2 + \frac{D}{r^2(1 + w^2 + \ell r^{-2})}t^2$$
for all $t \in [0,T_0]$ and \eqref{energy_upper} follows.

Next, we establish the lower bound on $T_0$. A brief calculation shows that $\ddot{\mcR}(t) \geq 0$, which implies
$$\dot{\mcR}(t) \geq \dot{\mcR}(0) = \frac{w}{\sqrt{1 + w^2 + \ell r^{-2}}}$$
for $t \in [0,T_0].$
Integrating over $[0,T_0]$ produces
$\mcR(T_0) - r \geq \frac{w}{\sqrt{1 + w^2 + \ell r^{-2}}} T_0$
and since $w < 0$ we find
\begin{eqnarray*}
T_0  & \geq & \frac{\sqrt{1 + w^2 + \ell r^{-2}}}{w} (\mcR(T_0) - r)\\
& \geq & \frac{\sqrt{1 + w^2 + \ell r^{-2}}}{w} (\mcR_+ - r)\\
& = & r\frac{\sqrt{1 + w^2 + \ell r^{-2}}}{\vert w \vert} \left (1 - \sqrt{\frac{D}{r^2w^2 + D}} \right)\\
& > & r \left (1 - \sqrt{\frac{D}{r^2w^2 + D}} \right).
\end{eqnarray*}
Of course, since the right side is nonnegative, this bound further implies $T_0 > 0$.

Finally, we show that $T_0 < \infty$. For the sake of contradiction, assume $T_0 =\infty$. Then, by \eqref{Rdec} we have 
$\mcR(t) \leq r$ for all $t \geq 0$. Using $m(t,\mcR(t)) \geq 0$ and \eqref{energy_lower}, we find 
$$ \dot{\mcW}(t) \geq \frac{\ell}{\mcR(t)^{3} \sqrt{1 + \mcW(t)^2 + \ell \mcR(t)^{-2}}} \geq \frac{\ell}{r^3\sqrt{1 + w^2 + \ell r^{-2}}}$$
and hence for all $t \geq 0$
$$\mcW(t) \geq \frac{\ell}{r^3\sqrt{1 + w^2 + \ell r^{-2}}}t + w.$$
Taking $t > \frac{-w r^3\sqrt{1 + w^2 + \ell r^{-2}}}{\ell} > 0$ implies $\mcW(t) > 0$, thus contradicting the assumption that $T_0 = \infty$, and we conclude that $T_0$ must be finite. In particular, the upper bound
$$ T_0 \leq \frac{-w r^3\sqrt{1 + w^2 + \ell r^{-2}}}{\ell}$$
follows from this argument.
The positivity of $\dot{\mcW}(t)$ for all $t >0$ further implies the uniqueness of $T_0$ and the proof is complete.
\end{proof}

With the characteristics well-understood, our other lemma provides lower bounds on the charge density and electric field in terms of the total mass and the position of the particle on the support of $f_0$ that is furthest from the origin.
The proof is identical to that of the non-relativistic system and can be found in \cite{BCP}.

\begin{lemma}
\label{L3}
Let $f(t,r, w, \ell)$ be a spherically-symmetric solution of \eqref{RVM} with associated charge density $\rho(t,r)$ and electric field $E(t,x)$, and let $\left (\mcR(t, 0, r, w, \ell), \mcW(t, 0, r, w, \ell), \mcL(t, 0, r, w, \ell) \right)$ be a characteristic solution of \eqref{charang}.  
If at some $T \geq 0$ we have
$$ \sup_{(r, w, l) \in S(0)} \mcR(T, 0, r, w, \ell) \leq B,$$ then
$$\Vert \rho(T) \Vert_\infty \geq \frac{3\Vert f \Vert_1}{4\pi B^3}$$
and
$$\Vert E(T) \Vert_\infty \geq \frac{\Vert f \Vert_1}{B^2}.$$
\end{lemma}

\section{Proofs of the main results}
\label{proofs}

\subsection{Initial data}
We will use differing initial datum, but with similar structure, to prove each of the main results.
Let $H:[0,\infty) \to [0,\infty)$ be any function satisfying
$$\int_{\mathbb{R}^3} H(\vert u \vert^2) \ du = \frac{3}{4\pi}$$
with $supp(H) \subset [0,1]$.
We rescale this function for any $\eps \in (0,1)$ by defining
$$H_\eps(\vert u \vert^2) = \frac{1}{\eps^3} H\left (\frac{\vert u \vert^2}{\eps^2} \right )$$
so that 
$$\int_{\mathbb{R}^3} H_\eps(\vert u \vert^2) \ du = \frac{3}{4\pi}$$
and
$supp(H_\eps) \subset [0,\eps^2].$
Further, for every $\eps > 0$, $x,p \in \mathbb{R}^3$, and $a_0>0$ define
$$h_\eps(x,p) = H_\eps \left ( \left \vert \frac{x}{\eps^2} + a_0 p \right \vert^2 \right ).$$
It follows that
\[
\int_{\mathbb{R}^3} h_\eps(x,p) \ dp = \frac{3}{4\pi a_0^3}
\]
for every $x \in \mathbb{R}^3$. 
We also choose a cut-off function $\phi \in C^\infty\left ((0,\infty); [0,1]\right )$ satisfying
\[
\left \{
\begin{array}{ll}
& \phi(r) = 0 \quad \mathrm{for} \quad r \not\in [a_0 - \epsilon^3, a_0 + \epsilon^3],\\
& \phi(r) = 1 \quad \mathrm{for} \quad r \in (a_0 - \frac{1}{2}\epsilon^3, a_0 + \frac{1}{2}\epsilon^3).
\end{array}
\right.
\]
With this, we may define the class of initial data within the proofs of Theorems \ref{T1} and \ref{T2}.
For any $a_0 > 0$, $\eps > 0$, and $M > 0$, we let
\begin{equation}
\label{finit1}
\mathring{f}_1(x,p) = h_\eps(x,p) \phi(\vert x \vert)
\end{equation}
and
\begin{equation}
\label{finit1}
\mathring{f}_2(x,p) = M \frac{\mathring{f}_1(x,p)}{\Vert \mathring{f}_1 \Vert_1}.
\end{equation}
Finally, we take $f_0 = \mathring{f}_1$ in the proof of Theorem \ref{T1}, and $f_0 = \mathring{f}_2$ in the proof of Theorem \ref{T2}.
Further, we note that the total mass will be chosen in the proof of Theorem \ref{T1} via the initial data, while this quantity is given within Theorem \ref{T2}.

Though we will define the parameters $a_0$, $\eps$, and $M$ differently within each proof, the data will share some common features.
From the upper bound on the support of $H_\eps$, we have on the support of $f_0(x,p)$ the inequality
$$\left \vert \frac{x}{\eps^2} + a_0 p \right \vert^2 < \eps^2.$$
Using the angular coordinates of \eqref{ang} and the identity $\vert p \vert^2 = w^2 + \ell r^{-2}$ this becomes 
\begin{equation}
\label{suppcond}
\left (\frac{r}{\eps^2} + a_0 w \right )^2 + \ell \left (\frac{a_0}{r} \right)^2 < \eps^2
\end{equation}
for every $(r,w,\ell) \in S(0)$.
Additionally, we have
\begin{equation}
\label{suppcond2}
a_0-\epsilon^3< r<a_0+\epsilon^3
\end{equation}
on the support of $f_0$.
Upon performing a translation, we find
$$ \int_{\mathbb{R}^3} f_0(x,p) \ dp = \left (\int_{\mathbb{R}^3} h_\eps(x,p) \ dp \right) \phi (|x|) = \frac{3}{4\pi a_0^3} \phi(|x|),$$ 
and it follows that the initial charge density $\rho_0=\int f_0\,dp$ satisfies
\begin{equation}
\label{boundrho}
\rho_0(r)\leq \frac{3}{4\pi a_0^3}, \qquad\forall r > 0
\end{equation}
and
\begin{equation}
\label{equalrho}
\rho_0(r) = \frac{3}{4\pi a_0^3},\qquad \text{for }r \in \left [a_0 - \frac{1}{2}\eps^3, a_0 + \frac{1}{2}\eps^3 \right ].
\end{equation}
Notice that~\eqref{suppcond} also implies
\begin{equation}
\label{suppcondl}
\ell  < \left (\frac{r}{a_0} \right )^2 \eps^2
\end{equation}
on $S(0)$.
Additionally, the support condition~\eqref{suppcond} further yields $\left \vert r + \eps^2a_0 w \right  \vert < \eps^3$ and thus \eqref{suppcond2} implies 
$$-\frac{1}{\eps^2} - \frac{2\eps}{a_0} < w <  -\frac{1}{\eps^2} + \frac{2\eps}{a_0}$$ on $S(0)$. 
Hence, all particles possess an initial radial velocity belonging to this interval.
Finally, we remark that since these initial data, $\mathring{f}_1$ and $\mathring{f}_2$ are spherically-symmetric, they must give rise to global-in-time, spherically-symmetric solutions of \eqref{RVM}.

\subsection{Proof of Theorem~\ref{T1}}
To prove the first result, the parameter $a_0$ will be fixed and we may choose $T = \mathcal{O}(a_0)$ and $\eps$ sufficiently small so that particles are quickly concentrated near the origin and obtain radial positions as small as one desires, thereby causing the density and field to become arbitrarily large at time $T > 0$.

\begin{proof}
Let $C_1, C_2>0$ be given, define the constant
$$a_0 = \left (\frac{32}{C_1} \right )^{1/3}$$
%
%
and set $$T = a_0 - 9\eps^2.$$
Throughout, we will take $\eps \in (0,1)$ sufficiently small, and in particular, choose $\eps < \frac{a_0}{9}$ to guarantee $T > 0$.

With this, \eqref{suppcond2} and \eqref{boundrho} imply that the total mass obeys the following upper bound for $\eps$ sufficiently small
\begin{eqnarray*}
M &=& \int_{\mathbb{R}^3} \rho_0(x) \ dx=4\pi \int_{a_0-\eps^3}^{a_0+\eps^3}\rho_0(r) r^2\,dr\\
& \leq & \frac{1}{a_0^3} \left [ (a_0+\eps^3)^3 - (a_0 - \eps^3)^3 \right ]\\
& = & \frac{6\eps^3}{a_0} +  \frac{2\eps^9}{a_0^3}\\
& \leq &  \frac{8\eps^3}{a_0},
\end{eqnarray*}
while \eqref{equalrho} implies that $M$ has the following lower bound
\begin{eqnarray*}
M &\geq& 4\pi \int_{a_0-\frac{1}{2}\eps^3}^{a_0+\frac{1}{2}\eps^3}\rho_0(r) r^2\,dr\\
& \geq & \frac{1}{a_0^3} \left [ (a_0+\frac{1}{2}\eps^3)^3 - (a_0 - \frac{1}{2}\eps^3)^3 \right ]\\
& = & \frac{3\eps^3}{a_0} +  \frac{\eps^9}{4a_0^3}\\
& \geq & \frac{3\eps^3}{a_0}
\end{eqnarray*}
for $\eps$ sufficiently small. %
Thus, we find
\begin{equation}
\label{mass}
3a_0^{-1}\eps^3 \leq M \leq 8a_0^{-1} \eps^3.
\end{equation}
On $S(0)$, taking $\eps$ sufficiently small further implies
\begin{equation}
\label{rw}
\left \{
\begin{gathered}
\frac{1}{2}a_0 < a_0 - \eps^3 < r  < a_0 + \eps^3< \frac{3}{2}a_0\\
-\frac{3}{2}\eps^{-2} < -\eps^{-2} - \frac{2\eps}{a_0} < w < -\eps^{-2}+ \frac{2\eps}{a_0} < -\frac{1}{2} \eps^{-2}.
\end{gathered}
\right.
\end{equation}
Additionally, \eqref{rw} combined with \eqref{suppcondl} implies a uniform upper bound on the angular momentum on $S(0)$, namely
\begin{equation}
\label{l}
\ell < \left (\frac{3}{2} \right )^2 \eps^2 \leq 1
\end{equation}
for $\eps$ sufficiently small.

To prove the conclusions of the theorem at time zero, we first notice that by~\eqref{boundrho}
$$\Vert \rho(0) \Vert_\infty \leq \frac{3}{4\pi a_0^3} \leq C_1.$$
Similarly, due to \eqref{fieldang} and \eqref{suppcond2} the field satisfies $|E(0,x)| = 0$ for $ \vert x \vert < a_0 - \eps^3$, while for $\vert x \vert > a_0 + \eps^3$
$$\vert E(0,x) \vert \leq \frac{M}{r^2} \leq \frac{M}{a_0^2} \leq \frac{8\eps^3}{a_0^3} \leq \frac{8}{a_0^3}.$$
Finally, for $a_0 -\eps^3 \leq \vert x \vert \leq a_0 + \eps^3$, we have
$$\vert E(0,x) \vert \leq \frac{M}{r^2} \leq \frac{M}{\left(\frac{1}{2} a_0 \right)^2} = \frac{4M}{a_0^2} \leq \frac{32 \eps^3}{a_0^3} \leq \frac{32}{a_0^3}.$$
Hence, we find
$$\Vert E(0) \Vert_\infty \leq \frac{32}{a_0^3} \leq C_1.$$
Therefore, we merely need to establish the contrasting inequalities at time $T$ to complete the proof.

Since the trajectories of particle positions are convex, they must each attain a minimum, and we use this construction to create a uniform lower bound over $S(0)$ on the time until particles attain their minima. 
In order to exclude those particles in $S(0)$ with vanishing angular momentum, we define  $$S_+ = \{ (r, w, \ell) \in S(0): \ell > 0\}.$$
Then, using Lemma \ref{L1}, we find for each $(r,w,\ell) \in S_+$ a time $T_0(r,w,\ell)$ such that
$$T_0 > r \left (1 - \sqrt{\frac{D}{r^2w^2 + D}} \right) \geq r - \frac{\sqrt{D}}{\vert w \vert} $$
where
$$D = \ell + Mr\sqrt{1 + w^2 + \ell r^{-2}} > 0,$$
and
$$\dot{\mcR}(t) \leq 0 \quad \mathrm{for} \quad  t \in [0,T_0].$$
Using \eqref{mass}, \eqref{rw}, and \eqref{l} we find
\begin{equation}
\label{Dbound}
D \leq 1 + \frac{8\eps^3}{a_0} \left (\frac{3}{2} a_0 \right ) \sqrt{1 + \frac{9}{4} \eps^{-4} + \frac{4}{a_0^2}} \leq 1 + C\eps \leq 4
\end{equation}
for $\eps$ sufficiently small.
%
Estimating on $S_+$, we use \eqref{rw} in order to arrive at
$$ T_0 > r - \frac{\sqrt{D}}{\vert w \vert} \geq (a_0 - \eps^3) - 4 \eps^2 > a_0 - 9\eps^2 = T$$
for $\eps$ sufficiently small.

Therefore, $T \in [0, T_0)$ for every $(r,w,\ell) \in S_+$ and we apply Lemma \ref{L3} to find
$$\mcR(T)^2 \leq \left (r - \frac{\vert w \vert}{\sqrt{1 + w^2 + \ell r^{-2}}}T \right )^2 + \frac{D}{r^2(1 + w^2 + \ell r^{-2})}T^2 =: I + II.$$
Defining $g(x) = \frac{1 - x}{1+ 2x}$ and noting that $g'(0) = -3$ and $g''(x) > 0$, we find 
\begin{eqnarray*}
\frac{\vert w \vert}{\sqrt{1 + w^2 + \ell r^{-2}}} & \geq & \frac{\frac{1}{\eps^2} - \frac{2}{a_0}\eps}{\sqrt{1 + \left (\frac{1}{\eps^2} + \frac{2}{a_0}\eps \right)^{2} + \left(\frac{1}{2}a_0 \right)^{-2}}}\\
& = & \frac{1 - \frac{2}{a_0}\eps^3}{\sqrt{\left (1 + \frac{2}{a_0}\eps^3 \right)^{2} + C\eps^4}}\\
& \geq & \frac{1 - \frac{2}{a_0}\eps^3}{\sqrt{1 + \frac{8}{a_0}\eps^3 + \frac{16}{a_0^2}\eps^6}}\\
& = & \frac{1 - \frac{2}{a_0}\eps^3}{1 + \frac{4}{a_0}\eps^3 }\\
& = & g \left (\frac{2}{a_0}\eps^3\right)\\
& \geq & g(0) + g'(0) \frac{2}{a_0}\eps^3\\
& = & 1 - \frac{6}{a_0}\eps^3
\end{eqnarray*}
for $\eps$ sufficiently small.
Because of this and $T = a_0 - 9\eps^2 < a_0 - \eps^3 \leq r$, it follows that
\begin{eqnarray*}
I  & \leq & \left (r - \left (1-\frac{6}{a_0}\eps^3 \right) T \right)^2\\
& \leq & \left ( a_0 + \eps^3 - [a_0 - 9\eps^2] + \frac{6}{a_0}\eps^3[a_0 - 9\eps^2]\right )^2\\
& = & \left (9 \eps^2 + 7\eps^3 - \frac{54}{a_0} \eps^5 \right )^2\\
& \leq & 336 \eps^4
\end{eqnarray*}
for $\eps$ sufficiently small.
Additionally, using \eqref{rw}, \eqref{Dbound}, and $0 < T \leq a_0$ it follows that
$$II \leq \frac{D}{r^2w^2}T^2 \leq \frac{4}{\left (\frac{1}{2} a_0 \right )^2 \left (\frac{1}{2} \eps^{-2} \right )^2} a_0^2 \leq 64\eps^4.$$

Combining these esimates yields  
$$\mcR(T)^2 \leq 400\eps^4.$$
Since this provides a uniform bound on $\mcR(T)$ over the set $S_+$, we take the supremum over all such triples to find
$$\sup_{(r,w,\ell) \in S(0)} \mcR(T, 0, r, w,\ell) = \sup_{(r,w,\ell) \in S_+} \mcR(T, 0, r, w,\ell) \leq 20 \eps^2.$$

Finally, invoking Lemma \ref{L3}, the upper bound on spatial characteristics implies a lower bound on the charge density, and therefore using \eqref{mass}
$$ \Vert \rho(T) \Vert_\infty \geq \frac{3M}{4\pi \left ( 20\eps^2 \right)^3 } = \frac{C}{a_0\eps^3} \geq C_2$$
for $\eps$ sufficiently small.
The same lemma also provides a lower bound on the field so that 
$$\Vert E(T) \Vert_\infty \geq \frac{M}{(20\eps^2)^2} = \frac{C}{a_0 \eps} \geq C_2$$
for $\eps$ sufficiently small, and the proof is complete.
\end{proof}

\subsection{Proof of Theorem~\ref{T2}}
Unlike the first result, $T > 0$ will be given here and we may choose $a_0$ sufficiently large and $\eps$ sufficiently small so that particles are far enough from the origin that the initial large momenta they experience will concentrate them near $r = 0$ only at the given time $T$. As before, this behavior implies that the density and field become arbitrarily large at this time.

\begin{proof}
Let $C_1,C_2>0$ and $T > 0$ be given, define the constant
$$C_0 = \sqrt{C_1T},$$
and set $$M=C_1 \qquad \mathrm{and} \quad a_0 = T + 16C_0\eps.$$
As before, we will take $\eps \in (0,1)$ sufficiently small throughout the proof.
Since $a_0 > T$, we find the useful inequalities
\begin{equation}
\label{rw2}
\left \{
\begin{gathered}
\frac{1}{2}a_0 < a_0 - \eps^3 < r  < a_0 + \eps^3 < \frac{3}{2}a_0\\
-\frac{3}{2}\eps^{-2} < -\frac{1}{\eps^2} - \frac{2}{T}\eps < w < -\frac{1}{\eps^2} + \frac{2}{T}\eps < -\frac{1}{2} \eps^{-2}
\end{gathered}
\right.
\end{equation}
on $S(0)$ and for $\eps$ sufficiently small.
Additionally, \eqref{rw2} combined with \eqref{suppcondl} implies a uniform upper bound on the angular momentum on $S(0)$ for $\eps$ sufficiently small, namely
\begin{equation}
\label{l2}
\ell < \left (\frac{a_0 + \eps^3}{a_0} \right )^2 \eps^2 \leq 1.
\end{equation}

As in the proof of Theorem \ref{T1}, we must exclude those particles in $S(0)$ with vanishing angular momentum, and thus we again let  $$S_+ = \{ (r, w, \ell) \in S(0): \ell > 0\}$$
and estimate on $S_+$.
Because the enclosed mass satisfies $0 \leq m(t,r) \leq M = C_1$ for all $t, r \geq 0$, we use Lemma \ref{L1} to find for each $(r,w,\ell) \in S_+$ a time $T_0(r,w,\ell)$ such that
$$T_0 > r \left (1 - \sqrt{\frac{D}{r^2w^2 + D}} \right) \geq r - \frac{\sqrt{D}}{\vert w \vert} $$
where
$$D = \ell + C_1r\sqrt{1 + w^2 + \ell r^{-2}} > 0,$$
and
$$\dot{\mcR}(t) \leq 0 \quad \mathrm{for} \quad  t \in [0,T_0].$$
Using \eqref{rw2} and \eqref{l2} we find for $\eps$ sufficiently small
\begin{eqnarray*}
D & \leq & 1 + C_1 \left ( a_0 + \eps^3 \right ) \sqrt{1 + \frac{9}{4} \eps^{-4} + \left (\frac{1}{2} a_0 \right)^{-2}}\\
& \leq & 1 + C_1 \left ( a_0 + \eps^3 \right ) \left ( 2\eps^{-2} \right ) \\
& \leq & 1 + C_1 \left ( T + 16C_0\eps + \eps^3 \right ) \left ( 2\eps^{-2} \right ) \\
& \leq & 4C_0^2\eps^{-2}.
\end{eqnarray*}

Next, we use \eqref{rw2} to produce a lower bound on $T_0$ and this yields
$$T_0 > r - \frac{\sqrt{D}}{\vert w \vert} \geq a_0 - \eps^3 - \frac{2C_0\eps^{-1}}{\frac{1}{2} \eps^{-2}} = a_0 - \eps^3 - 4C_0 \eps = T + 12C_0 \eps - \eps^3 \geq T $$
for $\eps$ sufficiently small.
Therefore, $T \in [0, T_0)$ for every $(r,w,\ell) \in S_+$ and we apply Lemma \ref{L1} to find
$$\mcR(T)^2 \leq \left (r - \frac{\vert w \vert}{\sqrt{1 + w^2 + \ell r^{-2}}}T \right )^2 + \frac{D}{r^2(1 + w^2 + \ell r^{-2})}T^2 =: I + II.$$
Defining $g(x) = \frac{1 - x}{1+ 2x}$ and noting that $g'(0) = -3$ and $g''(x) > 0$, we find 
\begin{eqnarray*}
\frac{\vert w \vert}{\sqrt{1 + w^2 + \ell r^{-2}}} & \geq & \frac{\frac{1}{\eps^2} - \frac{2}{T}\eps}{\sqrt{1 + \left (\frac{1}{\eps^2} + \frac{2}{T}\eps \right)^{2} + \left(\frac{1}{2}a_0 \right)^{-2}}}\\
& = & \frac{1 - \frac{2}{T}\eps^3}{\sqrt{\left (1 + \frac{2}{T}\eps^3 \right)^{2} + C\eps^4}}\\
& \geq & \frac{1 - \frac{2}{T}\eps^3}{\sqrt{1 + \frac{8}{T}\eps^3 + \frac{16}{T^2}\eps^6 }}\\
& = & \frac{1 - \frac{2}{T}\eps^3}{1 + \frac{4}{T}\eps^3 }\\
& = & g \left (\frac{2}{T}\eps^3\right)\\
& \geq & g(0) + g'(0) \frac{2}{T}\eps^3\\
& = & 1 - \frac{6}{T}\eps^3
\end{eqnarray*}
for $\eps$ sufficiently small.
Because of this and $T = a_0 - 16C_0\eps < a_0 - \eps^3 < r$, it follows that
$$I  \leq \left (r - \left (1-\frac{6}{T}\eps^3 \right )T \right)^2 \leq  \left ( a_0 + \eps^3 - T + 6\eps^3 \right )^2 = \left ( 16C_0\eps + 7\eps^3 \right )^2 \leq 336C_0^2 \eps^2$$
for $\eps$ sufficiently small.
Additionally, using \eqref{rw} and $0 < T \leq a_0$ it follows that
$$II \leq \frac{D}{r^2w^2}T^2 \leq \frac{4C_0^2 \eps^{-2}}{\left (\frac{1}{2} a_0 \right )^2 \left (\frac{1}{2} \eps^{-2} \right )^2} a_0^2 \leq 64C_0^2 \eps^2.$$

Combining these esimates yields  
$$\mcR(T)^2 \leq 400C_0^2\eps^2.$$
Since this provides a uniform bound on $\mcR(T)$ over the set $S_+$, we take the supremum over all such triples to find
$$\sup_{(r,w,\ell) \in S(0)} \mcR(T, 0, r, w,\ell) = \sup_{(r,w,\ell) \in S_+} \mcR(T, 0, r, w,\ell) \leq 20C_0 \eps.$$

Finally, invoking Lemma \ref{L3}, the upper bound on spatial characteristics implies a lower bound on the charge density and therefore
$$ \Vert \rho(T) \Vert_\infty \geq \frac{3C_1}{4\pi \left ( 20C_0\eps \right)^3 } = \frac{C}{\eps^3} \geq C_2$$
for $\eps$ sufficiently small.
The same lemma also provides a lower bound on the field so that 
$$\Vert E(T) \Vert_\infty \geq \frac{C_1}{(20C_0\eps)^2} = \frac{C}{\eps^2} \geq C_2$$
for $\eps$ sufficiently small, and the proof is complete.
\end{proof}


\begin{thebibliography}{99}



\bibitem{BCP} Ben-Artzi, J., Calogero, S., and Pankavich, S., Arbitrarily large solutions of the Vlasov-Poisson system. SIAM Journal on Mathematical Analysis (to appear), preprint - arXiv:1708.02307


\bibitem{Glassey} Glassey, R. The Cauchy Problem in Kinetic Theory.  SIAM: 1996.

\bibitem{GPS}  Glassey, R., Pankavich, S., and Schaeffer, J., Decay in Time for a One-Dimensional, Two Component Plasma. Math. Meth Appl. Sci. {\bf 2008}, 31:2115-2132.

\bibitem{GPS2} Glassey, R., Pankavich, S., and Schaeffer, J., On long-time behavior of monocharged and neutral plasma in one and
one-half dimensions. Kinetic and Related Models {\bf 2009}, 2: 465-488

\bibitem{GPS3} Glassey, R., Pankavich, S., and Schaeffer, J., Large Time Behavior of the Relativistic Vlasov-Maxwell System in Low Space Dimension. Differential and Integral Equations {\bf 2010}, 23: 61-77 

\bibitem{GPS4} Glassey, R., Pankavich, S., and Schaeffer, J., Time Decay for Solutions to the One-dimensional Equations of Plasma Dynamics. Quarterly of Applied Mathematics {\bf 2010}, 68: 135-141

\bibitem{GS} Glassey, R. and Schaeffer, J., On symmetric solutions of the relativistic Vlasov-Poisson system. Comm. Math. Phys. {\bf 1985} 101(4): 459--473

\bibitem{Horst} Horst, E, Symmetric plasmas and their decay. Comm. Math. Phys. {\bf 1990}, 126:613-633.

\bibitem{IR} Illner, R. and Rein, G., Time decay of the solutions of the Vlasov-Poisson system in the plasma physical case. Math. Methods Appl. Sci. {\bf 1996}, 19:1409-1413.






\bibitem{Yang} Yang, D., Growth estimates and uniform decay for the Vlasov Poisson system. Mathematical Methods in the Applied Sciences {\bf 2017}, DOI: 10.1002/mma.4356.

\bibitem{Young} Young, B., Landau damping in relativistic plasmas. J. Math. Phys. {\bf 2016}, 57(2), 021502.

\bibitem{Young2} Young, B., On linear {L}andau damping for relativistic plasmas via {G}evrey regularity. J. Diff. Eqns. {\bf 2015}, 259 (7): 3233--3273.

\end{thebibliography}
\end{document}